\documentclass[11pt]{article}
\usepackage[left=1in,top=1in,right=1in,bottom=1in,head=.1in,nofoot]{geometry}

\usepackage{setspace,url,bm,amsmath} 
\usepackage{subcaption}

\usepackage{titlesec} 
\titlelabel{\thetitle.\quad} 
\titleformat*{\section}{\bf\large\center} 

\usepackage{algorithm}
\usepackage{algpseudocode} 

\usepackage{graphicx} 
\usepackage{bbm}
\usepackage[export]{adjustbox}
\usepackage{latexsym}
\usepackage{mathrsfs}
\usepackage{amssymb}
\usepackage{authblk}
\usepackage{hyperref}
\usepackage[table]{xcolor}
\usepackage{tikz}
\usepackage{enumitem}
\usepackage{multirow}
\usepackage{rotating}

\usepackage{amsthm}
\theoremstyle{definition}

\newtheorem{proposition}{Proposition}
\newtheorem{lemma}{Lemma}
\newtheorem{remark}{Remark}

\newtheorem{example}{Example}

\newtheorem{corollary}{Corollary}

\usepackage{natbib} 
\bibpunct{(}{)}{;}{a}{}{,} 

\usepackage{etoolbox} 
\apptocmd{\sloppy}{\hbadness 10000\relax}{}{} 

\usepackage{color}
\usepackage{listings}

\title{Alternative statistical inference for the first normalized incomplete moment}
\author{Jiannan Lu\footnote{Apple Inc. Address for correspondence: 1 Apple Park Way, Cupertino, CA 95014, United States of America. Email: jiannan\_lu@apple.com},~
Peng Ding\footnote{
Department of Statistics, University of California, Berkeley.
}~ and Anqi Zhao\footnote{Fuqua School of Business, Duke University.}}
\date{\today}

\begin{document}
\maketitle

\begin{abstract}
This paper re-examines the first normalized incomplete moment, a well-established measure of inequality with wide applications in economic and social sciences. Despite the popularity of the measure itself, existing statistical inference appears to lag behind the needs of modern-age analytics. To fill this gap, we propose an alternative solution that is intuitive, computationally efficient, mathematically equivalent to the existing solutions for ``standard'' cases, and easily adaptable to ``non-standard'' ones. The theoretical and practical advantages of the proposed methodology are demonstrated via both simulated and real-life examples. In particular, we discover that a common practice in industry can lead to highly non-trivial challenges for trustworthy statistical inference, or misleading decision making altogether.
\end{abstract}

\section{Introduction}
\label{sec:intro}

\subsection{Background}

The past decade witnessed a surge in attention on the equality aspect of artificial intelligence, from recommendation system \citep{speicher2018unified} to machine learning model performance benchmarking \citep{lazovich2022measuring}, to A/B testing \citep{niculescu2021towards}. For example, \cite{polonioli2023ethics} stressed the need to account for all individuals in ML model evaluations, while \cite{lum2022biasing} pointed out the lack of rigorous measures (and the uncertainty quantifications thereof) for AI systems. 

From the economic and social science literature on income inequality \citep{champernowne1998economic}, one common measure is the share of income by the $p$th percentile \citep{butler1989using}, which finds wide applications in economics and beyond, is free of arbitrary utility functions or ``tuning parameters'' (e.g., $\epsilon$ in Atkinson's Index), and is among the most intuitive and actionable in the context of data mining. As an example, \cite{xie2016improving} investigated the correlation between streaming hours and user retention at Netflix, and concluded that ``users with low hours are more likely to be on the fence of cancellation, and more sensitive to product changes.'' Consequently, a measure of inequality can be share of streaming hours from the bottom 10\% active users.

Mathematically, let $X \in (0, \infty)$ be a positive-valued random variable with continuous probability density function $f(x)$ and cumulative distribution function 
$
F(x).
$
In addition, we let
$
\mu
$
and
$
\sigma^2
$
denote the mean and variance of $X.$ Worth noting that $f, F, \mu$ and $\sigma^2$ are \emph{unknown} but may be estimated from the observed data. Given known constant $p \in (0, 1),$ the share of $X$ at percentile $q=F^{-1}(p)$ is
\begin{equation}
\label{eq:estimand}
m(q) 
= \frac{\mathbb{E} \left( X 1_{\{X \le q\}} \right)}{\mu}
= \frac{\int_0^{q}sf(s)ds}{\int_0^{\infty}sf(s)ds}.
\end{equation}

\begin{remark}
\eqref{eq:estimand} is the Lorenz curve co-ordinate at $q,$ or as \cite{butler1989using} referred to, the ``first normalized incomplete moment.'' It is closely related to other measures of inequality, e.g., the Pietra Index \citep{eliazar2010measuring}
$
P = F(\mu) - m(\mu).
$
and the Gini Index \citep{lerman1984note}
$$
G = \int_0^\infty f(s) \left\{ \frac{s F(s)}{\mu} - m(s) \right\} ds. \qed
$$ 
\end{remark}

\begin{example}
\label{exa:1}
We consider two distributions that are widely leveraged in data mining. First, for the log-normal distribution, let $\Phi(x)$ denote the cumulative distribution function of the standard normal random variable. By definition, the cumulative distribution function of log-normal$(\mu, \sigma)$ is
$$
F_{\mathrm{LN}(\mu, \sigma)}(x) = \Phi\left(\frac{\mathrm{ln}x - \mu}{\sigma}\right),
$$
resulting
$$
q = \mathrm{exp} \left\{ \mu + \sigma \Phi^{-1}(p) \right\},
\quad
\mu = \mathrm{exp} \left(\mu + \sigma^2 / 2 \right).
$$
For log-normal, \eqref{eq:estimand} can be solved in closed-form \citep{butler1989using}, namely,
$$
m(q) = F_{\mathrm{LN}(\mu + \sigma^2, \sigma)}(q),
\quad
p \in (0, 1).
$$

Second, consider the exponential distribution. Let 
$
f_\lambda(x) = \lambda \mathrm{e}^{-\lambda x}
$ 
denote its probability density function, and 
$
F_{\mathrm{exp}(\lambda)} = 1 - \mathrm{e}^{-\lambda x},
$
denote its cumulative distribution function, which implies
$$
q = - \mathrm{ln}(1-p) / \lambda,
\quad
\mu = \lambda^{-1},
$$
and
$$
m(q) = 1 - \mathrm{e}^{-\lambda q} (1 + \lambda q),
\quad 
q \in (0, \infty).
$$ 
It is worth noting that the scaling factor $\lambda$ is not relevant here. \qed
\end{example}


\subsection{Literature review and contribution}

Although \eqref{eq:estimand} is well-established, there are few discussions around its statistical inference. For example, \cite{butler1989using} among many others focused on pre-specified values of $q,$ but in reality we need to estimate $q$ from the observed data -- given independently and identically samples
$
X_1, \ldots, X_n
$
from unknown distribution $F,$ how to construct statistically valid point estimate and confidence interval for $m(q)?$ 

From the literature, analytical \citep{beach1983distribution} and bootstrap-based solutions \citep{dupuis2006robust} had been proposed. Unfortunately, however, they lag behind the needs of modern-age analytics (e.g., on-line A/B experiments). On the one hand, closed-form solutions tend to be unintuitive and less applicable to ``non-standard'' scenarios in industry. For example, when the sample size $n$ is sufficiently large, should we treat
$
\hat q
$
as fixed and the true value for $q?$ On the other hand, re-sampling is computationally expensive or even infeasible for big data \citep{kleiner2014scalable}. 

In light of the opportunities, we make a three-fold contribution to the community. First, we increase awareness to the potential of these measures of inequality in modern-age analytics. Second, we create an \emph{alternative} solution that is intuitive, closed-form (therefore computationally efficient), mathematically equivalent to \cite{beach1983distribution} in ``standard'' settings, and generalizable to others in practice. Third, we reveal a somewhat surprising insight -- the often default and seemingly harmless exercise of assuming
$
q \approx \hat q
$
leads to highly non-trivial challenges for, or worse, outright misleading, statistical inference.

The remainder of the paper is organized as follows. Section \ref{sec:methodology} employs $m$-estimation \citep{newey1994large, stefanski2002calculus} to derive point estimate and confidence interval for $m,$ highlighting why closed-form solutions are preferred, both theoretically and practically. Section \ref{sec:example} illustrates the proposed methodology using simulated and real-life examples. Section \ref{sec:discussion} concludes and presents future directions. 

\section{Methodology}
\label{sec:methodology}

\subsection{Theory and practice}

For any fixed $p \in (0, 1),$ we estimate $q = F^{-1}(p)$ by its sample analogue
$
\hat q = 
X_{
\left(
\lfloor np \rfloor
\right),
}
$
resulting the following plug-in estimator for \eqref{eq:estimand}:
\begin{equation}
\label{eq:est2}
\hat m (\hat q) = \frac{\frac{1}{n}\sum_{i = 1}^n X_i 1_{\{X_i \le \hat q\}}}{\frac{1}{n}\sum_{i = 1}^n X_i}.
\end{equation}
Its asymptotic consistency is guaranteed by 
$
\hat q \overset{p}{\rightarrow} q,
$ 
and by applying the central limit theorem to both the denominator and the numerator of \eqref{eq:est2}, respectively.

To quantify the uncertainty of \eqref{eq:est2}, although it is possible to leverage bootstrap, closed-form solutions are much preferred. Unfortunately, however, deriving closed-form variance estimate is challenging, because the sample percentile $\hat q$ (as a non-linear function of $X_1, \ldots, X_n$) renders 
$$
X_1 1_{\{X_1 \le \hat q\}}, \ldots, X_n 1_{\{X_n \le \hat q\}}
$$ 
``non-trivially'' correlated.

To overcome this challenge, we re-write \eqref{eq:estimand} as
\begin{equation}
\label{eq:est-eq}
\mathbb{E} 
\left\{ 
X \mathrm{1}_{\{X \leq q\}} - m(q) X 
\right\} = 0.
\end{equation}
Solving the sample analogue of \eqref{eq:est-eq} yields \eqref{eq:est2}. To infer uncertainty, we let 
$$
Y = X \mathrm{1}_{\{X \leq q\}} - m(q) X,
\quad
Z = \mathrm{1}_{\{X \leq q\}} - p,
$$
and augment the estimating equation \eqref{eq:est-eq} as
$
\bm g(X, m, q) = 
\left( Y, Z \right)^T,
$
following the strategy of \cite{newey1984method}. It is trivial that $\mathbb{E} \bm g = \bm 0.$ We now present the main result of the paper.

\begin{proposition}
\label{prop:1}
Let 
$
\bm B
=
\begin{bmatrix}
\frac{\mathbb{E} \left\{ (Y - qZ)^2 \right\}}{\mu^2}
& - \frac{\mathbb{E}\left( YZ \right) + qf(q)\mathbb{E}\left( Z^2 \right)}{\mu f(q)}
\\
- \frac{\mathbb{E}\left( YZ \right) + qf(q)\mathbb{E}\left( Z^2 \right)}{\mu f(q)}
& 
\frac{p(1-p)}{f^2(q)}
\end{bmatrix}.
$
The joint asymptotic distribution of $\hat m(\hat q)$ and $\hat q$ is
\begin{equation}
\label{eq:est2-asym-dis}
\sqrt{n}
\begin{pmatrix} 
\hat m (\hat q) - m(q)
\\ 
\hat q - q
\end{pmatrix}
\overset{d}{\rightarrow} 
\bm N 
\left\{
\bm 0, \bm B
\right\}.
\end{equation}
\qed
\end{proposition}

\begin{proof}
The theory of $m-$estimation ensures that
\begin{equation}
\label{eq:est-eq-asym-var}
\sqrt{n} 
\begin{pmatrix} 
\hat m (\hat q) - m(q)
\\ 
\hat q - q
\end{pmatrix}
\overset{d}{\rightarrow} 
\bm N 
\left(
\bm 0, \bm B \bm M \bm B^T
\right),
\end{equation}
where
\begin{equation}
\label{eq:est-eq-asym-var-part1}
\bm B 
= 
\mathbb{E}
\left\{
\begin{pmatrix} 
\frac{\partial \bm g}{\partial m},
&
\frac{\partial \bm g}{\partial q}
\end{pmatrix}^{-1}
\right\}
= 
\begin{bmatrix}
- \mu & q f(q) \\
0 & f(q) 
\end{bmatrix}^{-1} 
=
-\frac{1}{\mu f(q)}
\begin{bmatrix}
f(q) & -qf(q) \\
0 & - \mu 
\end{bmatrix},
\end{equation}
and
\begin{equation}
\label{eq:est-eq-asym-var-part2}
\bm M
= \mathbb{E} (\bm g \bm g^T)
=
\begin{bmatrix}
\mathbb{E} \left( Y^2 \right) & \mathbb{E} \left( YZ \right) \\
\mathbb{E} \left( YZ \right) & \mathbb{E} \left( Z^2 \right)
\end{bmatrix}.
\end{equation}
Combining \eqref{eq:est-eq-asym-var}--\eqref{eq:est-eq-asym-var-part2},
\begin{align*}
\bm B \bm M \bm B^T
&= 
\frac{1}{\left(f(q)\mu\right)^2}
\begin{bmatrix}
f(q) & -qf(q) \\
0 & - \mu 
\end{bmatrix}
\begin{bmatrix}
\mathbb{E} \left( Y^2 \right) & \mathbb{E} \left( YZ \right) \\
\mathbb{E} \left( YZ \right) & \mathbb{E} \left( Z^2 \right)
\end{bmatrix} 
\begin{bmatrix}
f(q) & 0 \\
-qf(q) & - \mu 
\end{bmatrix} \\
&= 
\begin{bmatrix}
\frac{f(q) \mathbb{E}\left( Y^2 \right) -qf(q) \mathbb{E}\left( YZ \right)}{\left(f(q)\mu\right)^2} & \frac{f(q)\mathbb{E} \left( YZ \right) - qf(q)\mathbb{E}\left( Z^2 \right)}{\left(f(q)\mu\right)^2} \\
- \frac{\mu \mathbb{E}\left( YZ \right)}{\left(f(q)\mu\right)^2} & - \frac{\mu \mathbb{E}\left( Z^2 \right)}{\left(f(q)\mu\right)^2}
\end{bmatrix}
\begin{bmatrix}
f(q) & 0 \\
-qf(q) & - \mu 
\end{bmatrix} \\
&= \begin{bmatrix}
\frac{\mathbb{E} \left( Y^2 \right) -2q\mathbb{E}\left( YZ \right) + q^2 \mathbb{E}\left( Z^2 \right)}{ \mu^2}
& - \frac{\mu \mathbb{E}\left( YZ \right) + qf(q)\mu \mathbb{E}\left( Z^2 \right)}{f(q) \mu^2}
\\
- \frac{\mu \mathbb{E}\left( YZ \right) + qf(q)\mu \mathbb{E}\left( Z^2 \right)}{f(q)\mu^2}
& 
\frac{\left(\mu\right)^2 \mathbb{E}\left( Z^2 \right)}{f^2(q) \mu^2}
\end{bmatrix} \\
&= 
\begin{bmatrix}
\frac{\mathbb{E} \left\{ (Y - qZ)^2 \right\}}{\mu^2}
& - \frac{\mathbb{E}\left( YZ \right) + qf(q)\mathbb{E}\left( Z^2 \right)}{\mu f(q)}
\\
- \frac{\mathbb{E}\left( YZ \right) + qf(q)\mathbb{E}\left( Z^2 \right)}{\mu f(q)}
& 
\frac{\mathbb{E}\left( Z^2 \right)}{f^2(q)}
\end{bmatrix}.
\end{align*}
To complete the proof, note that
$
\mathbb{E}\left( Z^2 \right) = p(1-p)
$
because
$
\mathrm{1}_{\{X \leq q\}} \sim \mathrm{Bern}(p). \qed
$
\end{proof}

Now that the theoretical foundation is established, we can present the variance estimator for \eqref{eq:est2}. For 
$
i = 1, \ldots, n,
$
let
$$
\hat Y_i = X_i \mathrm{1}_{\{X_i \leq \hat q\}} - \hat m \left( \hat q \right)  X_i,
\quad
\hat Z_i = \mathrm{1}_{\{X_i \leq \hat q\}} - p.
$$
\eqref{eq:est2-asym-dis} implies the following plug-in variance estimator:
\begin{equation}
\label{eq:est2-asym-var-est}
\displaystyle
\hat{\mathbb{V}}_{\mathrm{a}} \left\{ \hat m (\hat q) \right\}
=
\frac{\sum_{i=1}^n ( \hat Y_i -  \hat q \hat Z_i )^2 }{\left(\sum_{i = 1}^n X_i\right)^2}.
\end{equation}

\begin{remark}
We highlight the following key benefits of the proposed alternative defined in \eqref{eq:est2-asym-dis}--\eqref{eq:est2-asym-var-est}:
\begin{enumerate}
\item \emph{Validity}: As a sanity check, \eqref{eq:est2-asym-dis} suggests the textbook result \citep{van2000asymptotic} in regard to the asymptotic variance of sample percentiles, namely 
$$
\hat{\mathbb{V}}_{\mathrm{a}} \left( \hat q \right)
=
\frac{p(1-p)}{nf^2(q)}.
$$
We will examine in greater details the validity of \eqref{eq:est2-asym-dis} in Section \ref{subsec:simu-1};

\item \emph{Simplicity}: In Appendix \ref{sec:equiv}, we prove that under the ``standard'' setting of
$
\hat q = 
X_{
\left(
\lfloor np \rfloor
\right)
},
$ 
\eqref{eq:est2-asym-dis} is mathematically equivalent to the asymptotic variance by \cite{beach1983distribution}. However, \eqref{eq:est2-asym-dis} is much more intuitive and simpler to compute;

\item \emph{Efficiency}: \eqref{eq:est2-asym-var-est} is more computationally efficient than bootstrap. Furthermore, after a single-pass to compute 
$
\hat q,
$
it can be directly transformed into distributed algorithms like MapReduce \citep{dean2008mapreduce}. 
\end{enumerate}
\end{remark}


\subsection{Treating $\hat q$ as fixed}

The main challenge in \eqref{eq:est2-asym-dis} is to properly account for the uncertainty of $\hat q.$ In practice where the sample size $n$ is large, we often consider 
$
\hat q
$
as fixed and the true value for $q.$ In order words, we assume
$
q \approx \hat q
$
to be known, in which case we similarly estimate \eqref{eq:estimand} by its finite-population analogue
\begin{equation}
\label{eq:est1}
\hat m (q) = \frac{\frac{1}{n}\sum_{i = 1}^n X_i 1_{\{X_i \le q\}}}{\frac{1}{n}\sum_{i = 1}^n X_i}, 
\end{equation}
whose asymptotic distribution of $\hat m_q$ is
\begin{equation}
\label{eq:est1-asym-dis}
\sqrt{n}\left\{ \hat m(q) - m(q) \right\}
\overset{d}{\rightarrow}
N \left(
0, 
\frac{\mathbb{E}Y^2}{\mu^2}
\right).
\end{equation}

\begin{proof}
With $q$ known, \eqref{eq:est-eq-asym-var} is reduced to
$$
\sqrt{n}\left\{ \hat m (q) - m(q) \right\}
 \overset{d}{\rightarrow} 
N 
\left[
0, \left\{ \mathbb{E} \left(\frac{dY}{dm} \right) \right\}^{-2} \mathbb{E}\left(Y^2\right)
\right]
=
N \left(
0, 
\frac{\mathbb{E} Y^2}{\mu^2}
\right).
$$

Alternatively, we can re-write \eqref{eq:est1} via first-order Taylor expansion \citep{deng2018applying}:
\begin{align*}
\hat m(q) - m(q)
& \approx 
\frac{\frac{1}{n}\sum_{i = 1}^n X_i 1_{\{X_i \le q\}} - \mathbb{E} \left( X 1_{\{X \le q\}} \right)}{\mu} \nonumber \\
&- \frac{\mathbb{E} \left( X 1_{\{X \le q\}} \right)}{\left( \mu \right)^2} 
\left(
\frac{1}{n}\sum_{i = 1}^n X_i - \mu
\right)
\nonumber \\
& = \frac{1}{n\mu} \sum_{i=1}^n 
\left\{
X_i 1_{\{X_i \le q\}} - \frac{\mathbb{E} \left( X 1_{\{X \le q\}} \right)}{\mu} X_i
\right\}
\nonumber \\
& = \frac{1}{n\mu} \sum_{i=1}^n Y_i.
\end{align*}
\eqref{eq:est1-asym-dis} follows immediately, because
$
Y_1, \ldots, Y_n
$ 
are i.i.d. with zero means.
\end{proof}

\subsection{A surprising insight}

To guide the trade-off between statistical rigor and computation efficiency, we quantify the difference between \eqref{eq:est2-asym-dis}, which appropriately treats $\hat q$ as a random variable, and \eqref{eq:est1-asym-dis}, which incorrectly treats $\hat q$ fixed. As it turns out, the latter can lead to misleading inference and decision-making. 

\begin{corollary}
\label{coro:1}
The difference between the variances of $\hat m (q)$ and that of $\hat m (\hat q)$ is
\begin{equation}
\label{eq:est1-est2-asym-var-diff}
\mathbb{V}_{\mathrm{a}} \left\{ \hat m (q) \right\}
- 
\mathbb{V}_{\mathrm{a}} \left\{ \hat m (\hat q) \right\}
=
\frac{q\left[
2m(q)\{1-m(q)\} \mu - q p(1-p)
\right]}{n \mu^2}.
\end{equation}
\end{corollary}

\begin{proof}
Recall 
$
Y = X \mathrm{1}_{\{X \leq q\}} - m(q)X
$
and
$
Z = \mathrm{1}_{\{X \leq q\}} - p.
$
By definition,
$$
\mathbb{E} Y = \mathbb{E} Z = 0
\quad
\mathbb{E}\left( Z^2 \right) = p(1-p).
$$
By \eqref{eq:est2-asym-dis} and \eqref{eq:est1-asym-dis}, we have
\begin{align}
\label{eq:eyz}
\mathbb{E} (YZ)
&= \mathbb{E}
\left(
X \mathrm{1}^2_{\{X \leq q\}} - m(q) X \mathrm{1}_{\{X \leq q\}}
\right) \nonumber \\
&= m(q)\{ 1 - m(q) \} \mu.
\end{align}
Therefore,
\begin{align}
\label{eq:diff-vars}
n \mu^2
\left[
\mathbb{V}_{\mathrm{a}} \left\{ \hat m (q) \right\}
- 
\mathbb{V}_{\mathrm{a}} \left\{ \hat m (\hat q) \right\}
\right]
&= \mathbb{E}\left( Y^2 \right) - \mathbb{E} \left\{ (Y - qZ)^2 \right\}
\nonumber \\
&= 2q  \mathbb{E}\left( YZ \right) - q^2 \mathbb{E}\left( Z^2 \right)
\nonumber \\
&= 2q m(q) \{ 1 - m(q) \} \mu - q^2 p(1-p).
\end{align}
\end{proof}

\begin{remark}
Intuitively, \eqref{eq:est1-est2-asym-var-diff} should be negative, because \eqref{eq:est1-asym-dis} ``leaves out'' the uncertainty around $\hat q.$ Surprisingly, the opposite is true for a wide range of distributions. Indeed, ignoring the uncertain around $\hat q$ often diminishes the negative correlation between $X$ and $1_{\{X \le \hat q\}},$ inflating the overall variance estimate. 

We illustrate this counter-intuitive phenomenon with two common distributions. First, for $X \sim \mathrm{Unif}(0, L),$ by \eqref{eq:estimand}
$$
q = p L,
\quad
\mu = L/2,
\quad
m(q) = q^2 / L^2 = p^2.
$$
Consequently,
\begin{align*}
\mathbb{V}_{\mathrm{a}} \left\{ \hat m (q) \right\}
- 
\mathbb{V}_{\mathrm{a}} \left\{ \hat m (\hat q) \right\}
&= \frac{4p^4(1-p)}{n} \geq 0.
\end{align*}
Second, for Exp$(\lambda),$
$
q = \mathrm{ln}(1-p) / \lambda,
$
or equivalently 
$
p = 1 - \mathrm{e}^{-\lambda q}
.$
Furthermore, from Example \ref{exa:1}
$$
\mu = \lambda^{-1},
\quad
m(q) = 1 - \mathrm{e}^{-\lambda q} (1 + \lambda q).
$$
To prove that \eqref{eq:est1-est2-asym-var-diff} is non-negative, we show that
$$
2m(1-m) \mu \geq q p(1-p),
$$
or equivalently,
\begin{equation}
\label{eq:est1-est2-asym-var-diff-exp}
h(t) = 
2
\left(
\mathrm{e}^t - 1 - t
\right)
(1 + t)
-
t 
\left(
\mathrm{e}^t - 1
\right)
\ge 0
\end{equation}
for all $t = \lambda q \in (0, \infty).$ To prove \eqref{eq:est1-est2-asym-var-diff-exp}, note that 
$$
h'(t) = t \mathrm{e}^t + 3 \mathrm{e}^t - 4t - 3
$$
and
$$
h''(t) = t \mathrm{e}^t + 4 \mathrm{e}^t - 4.
$$
Is it trivial that $h''(t) \ge 0$ for all $t \ge 0.$ Consequently, $h'(t) \ge h'(0) = 0$ for all $t \ge 0,$ which implies that $h(t) \ge h(0) = 0$ for all $t \ge 0.$ $\qed$
\end{remark}


\section{Illustrations}
\label{sec:example}

\subsection{Simulation studies}
\label{subsec:simu-1}

Let\footnote{For accessibility and reproduction purposes, computer programs can be found \href{https://github.com/jiannanlu/R-programs/blob/main/fairness.R}{here}.} $p = 0.75$ and consider the two distributions in Example \ref{exa:1}:
\begin{enumerate}
    \item $\mathrm{LN}(\mu, \sigma^2),$ where $(\mu, \sigma) \in \left\{(.4, .5), (-.3, 1), (.6, .5)\right\},$
    \item $\mathrm{Exp}(\lambda),$ where $\lambda \in \left\{.5, 1, 2\right\}.$
\end{enumerate}
The parameters are intended to mimic our production data-sets. 

For each case, we allow sample size $n \in \{2000, 5000, 10000\},$ and then:
\begin{enumerate}
\item Draw 
$
X_1, \ldots, X_n \stackrel{\mathrm{i.i.d.}}{\sim} F_{\mathrm{LN}(\mu, \sigma^2)},
$ 
\item Compute the point estimate $\hat m$ in \eqref{eq:est2}, and the three variance estimates:
$
\hat{\mathbb{V}}_{\mathrm{a}} \left\{ \hat m (\hat q) \right\}
$ 
in \eqref{eq:est2-asym-var-est},
$
\hat{\mathbb{V}}_{\mathrm{a}} \left\{ \hat m (q) \right\}
$
from \eqref{eq:est1-asym-dis} (which treats $\hat q$ as fixed), and 
$
\hat{\mathbb{V}}_{\mathrm{b}} \left\{ \hat m (\hat q) \right\}
$
which is based on generic non-parametric bootstrap (i.e., re-sampling with replacement $b = 200$ times);
\item Build 95\% confidence intervals (CIs) based on the above, and record whether they cover the ground truth $m,$ respectively.
\end{enumerate}
We repeat the above process $L=5000$ times, to approximate the true variance of $\hat m(\hat q),$ with the sampling variance of the point estimates retrieved. We then compute how far (on average) the numerous variance estimates deviate from said true variance, as well as their corresponding coverage rates. Table \ref{tab:simu} reports the results, based on which we draw two key take-aways: 
\begin{enumerate}
    \item Treating $q$ as fixed can lead to severely inflated variance estimates and over-covering. 
    \item Both our proposed closed-form solution and bootstrap achieve nominal (i.e., close to 95\%) coverage rates, and have small relative biases
\end{enumerate}

\begin{sidewaystable}[ph!]
\caption{Simulation results. The first two columns contain model parameters and sample sizes. The next two contain the true $m,$ and the (approximated) true sampling variance of $\hat m.$ The next three contain the relative bias of the three variance estimators, and the last three column contain their corresponding coverage rates.} 
\begin{center}
        \begin{tabular}{c | c | c | c | c | c | c | c | c | c}
			    \multicolumn{2}{c|}{case} & \multicolumn{2}{c|}{ground truth} & \multicolumn{3}{c|}{relative bias $(\mathbb{E} \hat{\mathbb{V}} - \mathbb{V}) / \mathbb{V}$} & \multicolumn{3}{c}{coverage rate of 95\% CI} \\
			    \hline
                 log-normal & $n$ & $m$ & $\mathbb{V}\{\hat m(\hat q) \}$ & $\hat{\mathbb{V}}_{\mathrm{a}} \left\{ \hat m (\hat q) \right\} $ & $\hat{\mathbb{V}}_{\mathrm{a}} \left\{ \hat m (q) \right\}$ & $\hat{\mathbb{V}}_{\mathrm{b}} \left\{ \hat m (\hat q) \right\}$ & $\hat{\mathbb{V}}_{\mathrm{a}} \left\{ \hat m (\hat q) \right\} $ & $\hat{\mathbb{V}}_{\mathrm{a}} \left\{ \hat m (q) \right\}$ & $\hat{\mathbb{V}}_{\mathrm{b}} \left\{ \hat m (\hat q) \right\}$ \\
                \hline
			     & 2000 & 0.57 & $1.50 \times 10^{-5}$ & 0.02\% & 1062.75\% & 1.91\% & 94.66\% & 100.00\% & 94.70\%  \\
			    $\mu = 0.4, \sigma = 0.5$ & 5000 & 0.57 & $6.10 \times 10^{-6}$ & -1.32\% & 1046.47\% & -0.66\% & 94.18\% & 100.00\% & 94.34\% \\
			    & 10000 & 0.57 & $3.02 \times 10^{-6}$ & -0.18\% & 1058.28\% & 0.03\% & 94.98\% & 100.00\% & 94.76\% \\
			    \hline
			     & 2000 & 0.37 & $7.33 \times 10^{-5}$ & 1.66\% & 199.79\% & 1.64\% & 94.80\% & 99.94\% & 94.74\% \\
			    $\mu = -0.3, \sigma = 1.0$ & 5000 & 0.37 & $3.09 \times 10^{-5}$ & -3.01\% & 185.12\% & -3.05\% & 94.46\% & 99.96\% & 94.32\% \\
			    & 10000 & 0.37 & $1.54 \times 10^{-5}$ & -2.11\% & 187.25\% & -2.00\% & 94.36\% & 99.94\% & 94.38\% \\
			    \hline
			     & 2000 & 0.57 & $1.52 \times 10^{-5}$ & -1.17\% & 1050.78\% & 0.53\% & 94.48\% & 100.00\% & 94.70\% \\
			    $\mu = 0.6, \sigma = 0.5$ & 5000 & 0.57 & $6.09 \times 10^{-6}$ & -1.13\% & 1048.58\% & -0.36\% & 95.06\% & 100.00\% & 95.04\% \\
			    & 10000 & 0.57 & $2.96 \times 10^{-6}$ & 2.05\% & 1083.78\% & 2.40\% & 95.24\% & 100.00\% & 95.16\%  \\
                \hline
                 exponential & $n$ & $m$ & $\mathbb{V}\{\hat m(\hat q) \}$ & $\hat{\mathbb{V}}_{\mathrm{a}} \left\{ \hat m (\hat q) \right\} $ & $\hat{\mathbb{V}}_{\mathrm{a}} \left\{ \hat m (q) \right\}$ & $\hat{\mathbb{V}}_{\mathrm{b}} \left\{ \hat m (\hat q) \right\}$ & $\hat{\mathbb{V}}_{\mathrm{a}} \left\{ \hat m (\hat q) \right\} $ & $\hat{\mathbb{V}}_{\mathrm{a}} \left\{ \hat m (q) \right\}$ & $\hat{\mathbb{V}}_{\mathrm{b}} \left\{ \hat m (\hat q) \right\}$ \\
                \hline
			     & 2000 & 0.40 & $4.06 \times 10^{-5}$ & 0.61\% & 378.89\% & 1.79\% & 95.14\% & 100.00\% & 95.20\%  \\
			    $\lambda = 0.5$ & 5000 & 0.40 & $1.61 \times 10^{-5}$ & 1.32\% & 382.73\% & 1.64\% & 94.84\% & 100.00\% & 94.80\% \\
			    & 10000 & 0.40 & $8.01 \times 10^{-6}$ & 1.84\% & 385.17\% & 2.00\% & 95.24\% & 99.98\% & 95.34\% \\
			    \hline
			     & 2000 & 0.40 & $4.05 \times 10^{-5}$ & 0.69\% & 379.57\% & 1.82\% & 95.08\% & 100.00\% & 95.08\% \\
			    $\lambda = 1.0$ & 5000 & 0.40 & $1.58 \times 10^{-5}$ & 3.30\% & 392.28\% & 3.86\% & 95.58\% & 100.00\% & 95.64\% \\
			    & 10000 & 0.40 & $8.08 \times 10^{-6}$ & 0.95\% & 380.73\% & 1.20\% & 95.24\% & 100.00\% & 95.02\% \\
			    \hline
			     & 2000 & 0.40 & $4.18 \times 10^{-5}$ & -2.31\% & 365.05\% & -1.17\% & 94.92\% & 99.98\% & 95.00\% \\
			    $\lambda = 2.0$ & 5000 & 0.40 & $1.64 \times 10^{-5}$ & -0.61\% & 373.20\% & -0.09\% & 94.80\% & 100.00\% & 94.82\% \\
			    & 10000 & 0.40 & $8.11 \times 10^{-6}$ & 0.58\% & 379.05\% & 0.95\% & 95.42\% & 100.00\% & 95.32\%  \\
			    \hline                
		\end{tabular}
	\end{center}
\label{tab:simu}
\end{sidewaystable}

To end this section, we compare the computation times of numerous methods. We draw i.i.d. samples
$
X_1, \ldots, X_n
$ 
from the same log-normal and exponential models, calculate the proposed variance estimate in \eqref{eq:est2-asym-var-est} and bootstrap (re-sampling with replacement $b = 200$ times) variance estimate and repeat this step for $L=100$ times (to account for fluctuations between runs), and report average (over the 100 repetitions) execution times in R, for both variance estimates. Table \ref{tab:add-simu} reports the results, which clearly indicate the superiority of our proposed closed-form solution \eqref{eq:est2-asym-var-est} -- bootstrap takes 200-1000 times of execution time, depending on the sample size and number of re-sampling chosen. 

\begin{table}[H]
\caption{Simulation times based on log-normals and exponentials.}
	\begin{center}
		\begin{tabular}{c | c | c | c}
			    \multicolumn{2}{c|}{case} & \multicolumn{2}{c}{average execution time (in milliseconds)} \\
			    \hline
                 log-normal & $n$ & $\hat{\mathbb{V}}_{\mathrm{a}} \left\{ \hat m (\hat q) \right\} $ & $\hat{\mathbb{V}}_{\mathrm{b}} \left\{ \hat m (\hat q) \right\}$ \\
                \hline
			     & 2000 & 0.50 & 147.42 \\
			    $\mu = 0.4, \sigma = 0.5$ & 5000 & 0.68 & 320.92 \\
			    & 10000 & 1.37 & 508.09 \\
			    \hline
			     & 2000 & 0.35 & 130.97 \\
			    $\mu = -0.3, \sigma = 1.0$ & 5000 & 0.78 & 288.19 \\
			    & 10000 & 1.58 & 505.26 \\
			    \hline
			     & 2000 & 0.61 & 133.60 \\
			    $\mu = 0.6, \sigma = 0.5$ & 5000 & 1.21 & 317.30 \\
			    & 10000 & 1.68 & 513.75 \\
			    \hline
                  exponential & $n$ & $\hat{\mathbb{V}}_{\mathrm{a}} \left\{ \hat m (\hat q) \right\} $ & $\hat{\mathbb{V}}_{\mathrm{b}} \left\{ \hat m (\hat q) \right\}$ \\
                \hline
			     & 2000 & 0.37 & 143.60 \\
			    $\lambda = 0.5$ & 5000 & 0.86 & 292.98 \\
			    & 10000 & 1.23 & 505.96 \\
			    \hline
			     & 2000 & 0.27 & 127.16 \\
			    $\lambda = 1.0$ & 5000 & 0.50 & 289.19 \\
			    & 10000 & 1.15 & 548.83 \\
			    \hline
			     & 2000 & 0.27 & 129.56 \\
			    $\lambda = 2$ & 5000 & 0.81 & 293.16 \\
			    & 10000 & 1.41 & 516.65 \\
\hline
        \end{tabular}
	\end{center}
\label{tab:add-simu}
\end{table}

\subsection{Empirical study}

We examine the Determinants of Wages Data \citep{kleiber2008applied}, a public data-set originated from US Census Bureau's 1988 Current Population Survey. We focus on \emph{wage} (US dollars per week) and \emph{smsa} (lives in a standard metropolitan statistical area). Table \ref{tab:cps} reports the estimated shares of wages from the bottom 75\% populations. Treating $\hat q$ as constant will lead to more than 400\% inflated variance estimates. Consequently, if we test the null hypothesis 
$$
m_\mathrm{urban}\left\{F^{-1}(0.75)\right\} = m_\mathrm{suburb}\left\{F^{-1}(0.75)\right\},
$$
the corresponding $t-$statistic drops from 2.59 ($p-$val = 0.01) to 1.22 ($p-$val = 0.11), leading to outright misleading decision-making. The results again highlight that trustworthy inference like \eqref{eq:est2-asym-var-est} is imperative.

\begin{table}[H]
\caption{Statistical inference for shares of wages from the bottom 75\%.}    
	\begin{center}
		\begin{tabular}{ c |c| c | c | c}
			    group & size & $\hat m$ & $\hat{\mathbb{V}}_{\mathrm{a}} \left( \hat m \right)$  & $\hat{\mathbb{V}}_{\mathrm{a}} \left( \hat m_q \right)$ \\
			    \hline
			    urban (smsa = ``yes'') & 20932 & 0.541 & $3.34 \times 10^{-6}$ & $1.96 \times 10^{-5}$ \\
			    \hline
			    suburb (sama = ``no'') & 7223 & 0.530 & $1.42 \times 10^{-5}$ & $6.01 \times 10^{-5}$ \\
			    \hline
		\end{tabular}
	\end{center}
\label{tab:cps}
\end{table}

\section{Concluding remarks}
\label{sec:discussion}
This paper proposed an alternative variance estimator for the first normalized incomplete moment. Our proposed methodology was statistically valid, computationally cheap, and easily adaptable. In particular, we discover that a common practice under the ``big data'' assumption can lead to highly non-trivial challenges for trustworthy statistical inference, or misleading decision making.

Our work suggests multiple future directions. First, we can potentially generalize our results to the joint inference of $\left\{\hat m(\hat q)\right\}_{q \in (0, \infty)},$ which converges to a gaussian process. Second, we can extend our methodology to other inequality measures (e.g., Gini Index), which are also closely related to incomplete moments. Third, it might be interesting to investigate under which conditions \eqref{eq:est1-est2-asym-var-diff} is non-negative.

\section*{Acknowledgments} 
The authors gratefully acknowledge their common thesis advisor at Harvard, Professor Tirthankar Dasgupta. JL thanks his colleagues at Apple, Devi Krishna, Shajeev Ramamoorthy, Jiulong Shan, Daphne Luong and Alexander Braunstein, for their continued encouragement and support. PD is partially supported by the US National Science Foundation Grant \#1945136. 

\appendix

\section{Mathematical equivalence under ``standard'' settings}
\label{sec:equiv}

As pointed out by \cite{brzezinski2013asymptotic}, when
$
\hat q = 
X_{
\left(
\lfloor np \rfloor
\right)
},
$ 
the asymptotic variance of 
$
\hat m (\hat q)
$
proposed by \cite{beach1983distribution} is
\begin{equation}
\label{eq:est-eq-asym-var-bd}
\mathbb{V}_\mathrm{bd} \left\{ \hat m (\hat q) \right\}
=
\frac{p}{N \mu^2}
\left[
\sigma^2 \frac{p \gamma^2}{\mu^2} + \epsilon^2 \left(1 - 2 \frac{p\gamma}{\mu} \right) + (1 - p) (q - \gamma)^2 - 2 \frac{p\gamma}{\mu}(q - \gamma) (\mu - \gamma)
\right],
\end{equation}
where
\begin{equation}
\label{eq:bd-gamma-definition}
\gamma = \mathbb{E} (X \mid X \le q) = \frac{m(q)\mu}{p},
\end{equation}
and
\begin{equation}
\label{eq:bd-eps-definition}
\epsilon^2 = \mathbb{V} (X \mid X \le q) = \frac{
\mathbb{E}
\left(
X^2 \mathrm{1}_{\{X \leq q\}}
\right)
}
{p}  
- \gamma^2
\end{equation}
are the conditional mean and variance of $X$ given $X \le q,$ respectively. Related asymptotic analyses appear in \cite{bickel1965some} and \cite{stigler1973asymptotic}.

\begin{lemma}
\label{lemma:1}
The following holds:
\begin{equation}
\label{eq:equiv-1}
\mu^2 m^2 - (2m-1) p\gamma^2 + p(1-p)q^2 - 2q m(1-m) \mu
=
p(1 - p) (\gamma - q)^2 - 2pm(q - \gamma) (\mu - \gamma).
\end{equation}
\end{lemma}

\begin{proof}
To simplify notation, we denote $m(q)$ by $m.$ Suggested by \eqref{eq:bd-gamma-definition},
$$
\frac{p\gamma}{\mu} = m,
\quad
p\gamma = m \mu.
$$
Consequently,
\begin{align}
\label{eq:equiv-3}
- 2q m(1-m) \mu
&=
-2qm\mu + 2mqm\mu \nonumber \\
&=
-2qp\gamma + 2mqp\gamma \nonumber \\
&=
-2pq\gamma + 2p^2 q\gamma - 2pmq\mu + 2pmq\gamma \nonumber \\
&=
- 2p(1-p) q\gamma - 2pm(q\mu - q\gamma).
\end{align}
First, adding $\mu^2 m^2$ to the left hand side of \eqref{eq:equiv-3}, and
$$
2pm\gamma\mu - p^2\gamma^2 (= \mu^2 m^2)
$$
to its right hand side,
\begin{align}
\label{eq:equiv-4}
\mu^2 m^2 - 2q m(1-m) \mu
&=
2pm\gamma\mu - p^2\gamma^2 - 2p(1-p) q\gamma - 2pm(q\mu - q\gamma) \nonumber \\
&=
- p^2\gamma^2 - 2p(1-p)q\gamma - 2pm(q\mu - \gamma \mu - q\gamma).
\end{align}
Next, subtracting 
$
(2m-1) p\gamma^2
$
to both sides of \eqref{eq:equiv-4}, 
\begin{align}
\label{eq:equiv-5}
\mu^2 m^2 - (2m-1) p\gamma^2 - 2q m(1-m) \mu
&=
- p^2\gamma^2 - 2p(1-p)q\gamma - 2pm(q\mu - \gamma \mu - q\gamma) - (2m-1) p\gamma^2 \nonumber \\
&= 
p(1-p) \gamma^2 - 2p(1-p) q\gamma - 2pm(q\mu - \gamma\mu - q\gamma + \gamma^2) \nonumber \\
&=
p(1 - p) (\gamma^2-2q\gamma) - 2pm(q - \gamma) (\mu - \gamma).
\end{align}
Finally, to complete the proof of \eqref{eq:equiv-1}, simply add 
$
p(1-p)q^2
$
to both sides of \eqref{eq:equiv-5}.
\end{proof}

\begin{proposition}
\label{prop:equiv}
\cite{beach1983distribution}'s formula \eqref{eq:est-eq-asym-var-bd} is mathematically equivalent to \eqref{eq:est2-asym-dis}.
\end{proposition}

\begin{proof}
On the one hand, by \eqref{eq:bd-eps-definition},
\begin{align}
\label{eq:ey2}
\mathbb{E} Y^2
&=
\mathbb{E}
\left(
m^2 X^2 - 2m X^2 \mathrm{1}_{\{X \leq q\}} + X^2 \mathrm{1}^2_{\{X \leq q\}}
\right) \nonumber \\
&=
\mathbb{E}
\left\{
m^2 X^2 - (2m -1) X^2 \mathrm{1}_{\{X \leq q\}}
\right\} \nonumber \\
&=
\sigma^2 m^2 + \mu^2 m^2 - (2m-1) \mathbb{E}\left(X^2 \mathrm{1}_{\{X \leq q\}} \right) \nonumber \\
&=
\sigma^2 m^2 + \mu^2 m^2 - (2m-1)p (\epsilon^2 + \gamma^2).
\end{align}
Using \eqref{eq:ey2} and \eqref{eq:eyz}, we re-write \eqref{eq:est2-asym-dis} as
\begin{align}
\label{eq:est2-asym-dis-2}
\mathbb{V}_\mathrm{a} \left\{ \hat m (\hat q) \right\}
&=
\frac{1}{N \mu^2}
\mathbb{E} \left(Y^2 - 2qYZ + q^2 Z^2 \right) \nonumber \\
&= 
\frac{1}{N \mu^2}
\left\{
\sigma^2 m^2 + \mu^2 m^2 - (2m-1)p (\epsilon^2 + \gamma^2)
+
q^2 p(1-p) - 2q m(1-m) \mu
\right\} \nonumber \\
&=
\frac{1}{N \mu^2}
\left\{
\sigma^2 m^2 + (1-2m)p \epsilon^2 + \mu^2 m^2 - (2m-1)p \gamma^2
+
q^2 p(1-p) - 2q m(1-m) \mu
\right\}.
\end{align}
On the other hand, we re-write \eqref{eq:est-eq-asym-var-bd} as
\begin{equation}
\label{eq:est-eq-asym-var-bd-2}
\mathbb{V}_\mathrm{bd} \left\{ \hat m (\hat q) \right\}
=
\frac{1}{N \mu^2}
\left[
\sigma^2 m^2 + (1 - 2m)p\epsilon^2  + p(1 - p) (q - \gamma)^2 - 2pm(q - \gamma) (\mu - \gamma)
\right].
\end{equation}
By Lemma \ref{lemma:1}, \eqref{eq:est-eq-asym-var-bd-2} = \eqref{eq:est2-asym-dis-2}, and the proof is complete.
\end{proof}

\newpage
\bibliographystyle{apalike}
\bibliography{reference} 

\end{document}